\documentclass[a4paper]{amsart}
\usepackage{amsmath,amssymb,epsfig,times, subfigure}

\newcommand{\e}{\epsilon}

\makeatletter \@addtoreset{equation}{section} \makeatother

\newtheorem{theorem}{Theorem}[section]

\newcommand{\ra}{\rightarrow}
\begin{document}
\title[Multiscale expansion of KdV]{
Numerical study of a multiscale expansion of the Korteweg de Vries 
equation and Painlev\'e-II equation}
\begin{abstract}
    The Cauchy problem for the Korteweg de Vries (KdV) equation with 
     small dispersion of order $\e^2$, $\e\ll 1$, 
    is characterized by the appearance of a zone of rapid modulated oscillations.
    These oscillations are approximately described by the elliptic solution of KdV where the amplitude,
     wave-number and frequency are not constant but evolve according 
     to the Whitham equations. Whereas the difference between the KdV 
     and the asymptotic solution decreases as $\epsilon$ in the 
     interior of the Whitham oscillatory zone, it is known to be only of order 
     $\epsilon^{1/3}$ near the leading edge of this 
     zone. To obtain a more accurate description 
     near the leading edge of the oscillatory zone we present a multiscale expansion of the 
     solution of KdV  in terms of the Hastings-McLeod solution of the 
     Painlev\'e-II equation. We show numerically that the resulting 
     multiscale  solution approximates the KdV solution, in the small dispersion limit,
 to the order           
     $\epsilon^{2/3}$. 
\end{abstract}
\author[T. Grava]{T. Grava}
 \address{SISSA, via Beirut 2-4, 34014 Trieste, Italy} 
 \email{grava@fm.sissa.it}
\author[C. Klein]{C. Klein}
 \address{Max Planck Institute for Mathematics in the Sciences} 
 \email{klein@mis.mpg.de}
 \thanks{We thank B.~Dubrovin and J.~Frauendiener for helpful 
discussions and hints. We acknowledge support by the MISGAM program 
of the European Science Foundation. TG  acknowledges support by the RTN ENIGMA and  Italian 
COFIN 2004 ``Geometric methods in the theory of nonlinear waves and their applications''. The authors wish to thank the referees for the improvements suggested to the manuscript.}

\maketitle
\section{Introduction}
The mathematically rigorous study of the small dispersion limit of the 
Korteweg de Vries (KdV) equation 
\begin{equation}
    u_{t}+6uu_{x}+\epsilon^{2}u_{xxx}=0, \quad  \epsilon\ll 1,
    \label{p22}
\end{equation}
with smooth initial data $u_0(x)$  was initiated in the works 
of Lax-Levermore \cite{LL}, which stimulated intense research
both numerically and analytically on the problem.
\begin{figure}[!htb]
 \centering\epsfig{figure=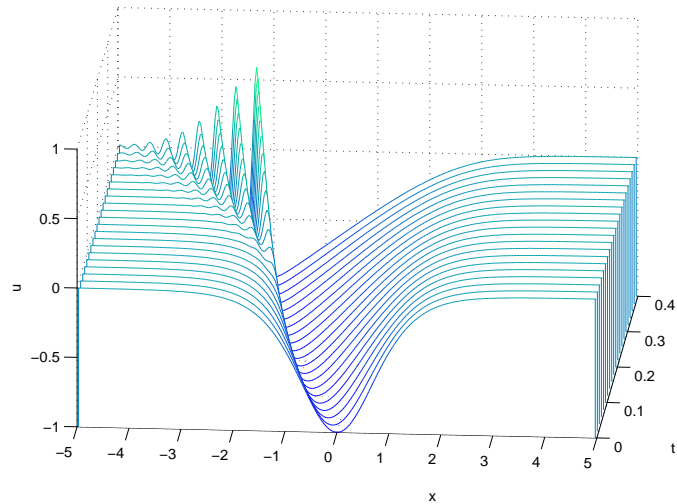, width=9.5cm}
 \caption{Numerical solution of the KdV equation for 
 the initial data 
 $u_0(x)=-\mbox{sech}^2x$  and $\epsilon=0.1$.}
 \label{fig1}
 \end{figure}
 The solution of the
Cauchy problem of the KdV equation in the small dispersion 
limit is characterized by the appearance of a zone of rapid 
oscillations of frequency of order $1/\epsilon$, see for instance Fig.~\ref{fig1}.

These oscillations are formed in the strong nonlinear regime and they
 have been analytically described in terms of 
elliptic functions, and in the general case in terms of 
theta functions  in \cite{V2},
\cite{DVZ}; the evolution in time of the oscillations was studied in \cite{FRT1}.
These results give a good asymptotic description of the oscillations 
only near the center
of the oscillatory zone (see Fig.~\ref{fig4}). 
\begin{figure}[!htb]
\centering
\epsfig{figure=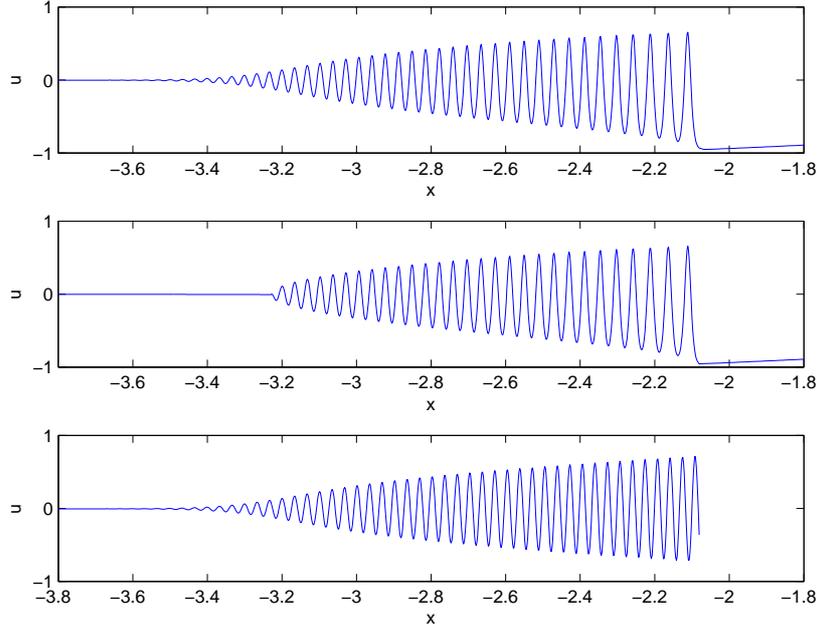, width=1.00\textwidth}
\caption{We plot for $u_{0}=-\mbox{sech}^{2}x$, 
$t=0.4$ and  $\epsilon=10^{-2}$  from top to bottom:  
1) the numerical solution of KdV;  2)  the asymptotic formula (\ref{elliptic}) in terms of elliptic functions  and the Hopf solution;  3)  the multiscale solution where the envelope of the oscillations is given by a solution to the Painlev\'e-II equation.}
\label{fig4}
\end{figure}
\begin{figure}[!htb]
\centering
\mbox{\subfigure[]
{\epsfig{figure=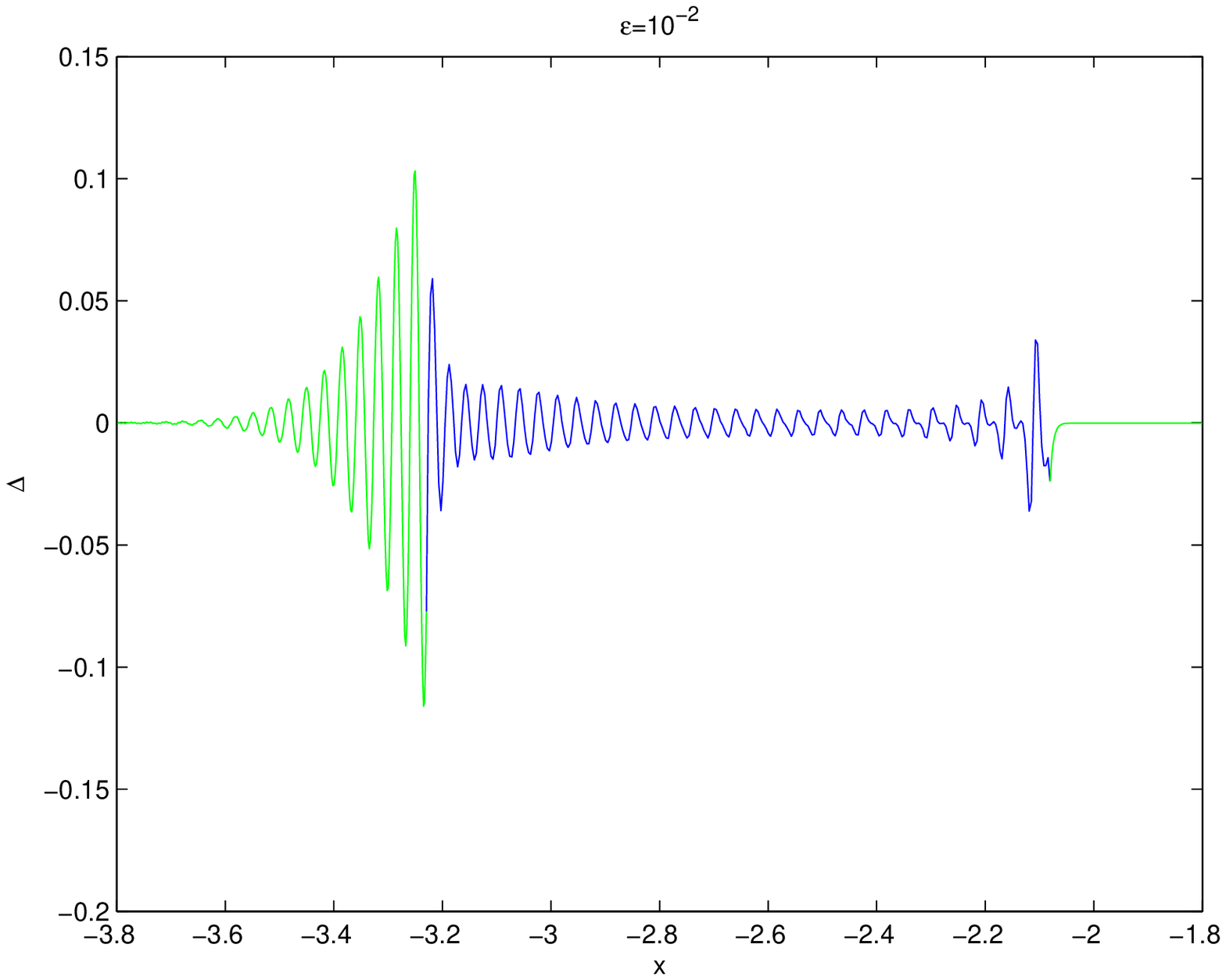, width=.55\textwidth}}
\subfigure[]
{\epsfig{figure=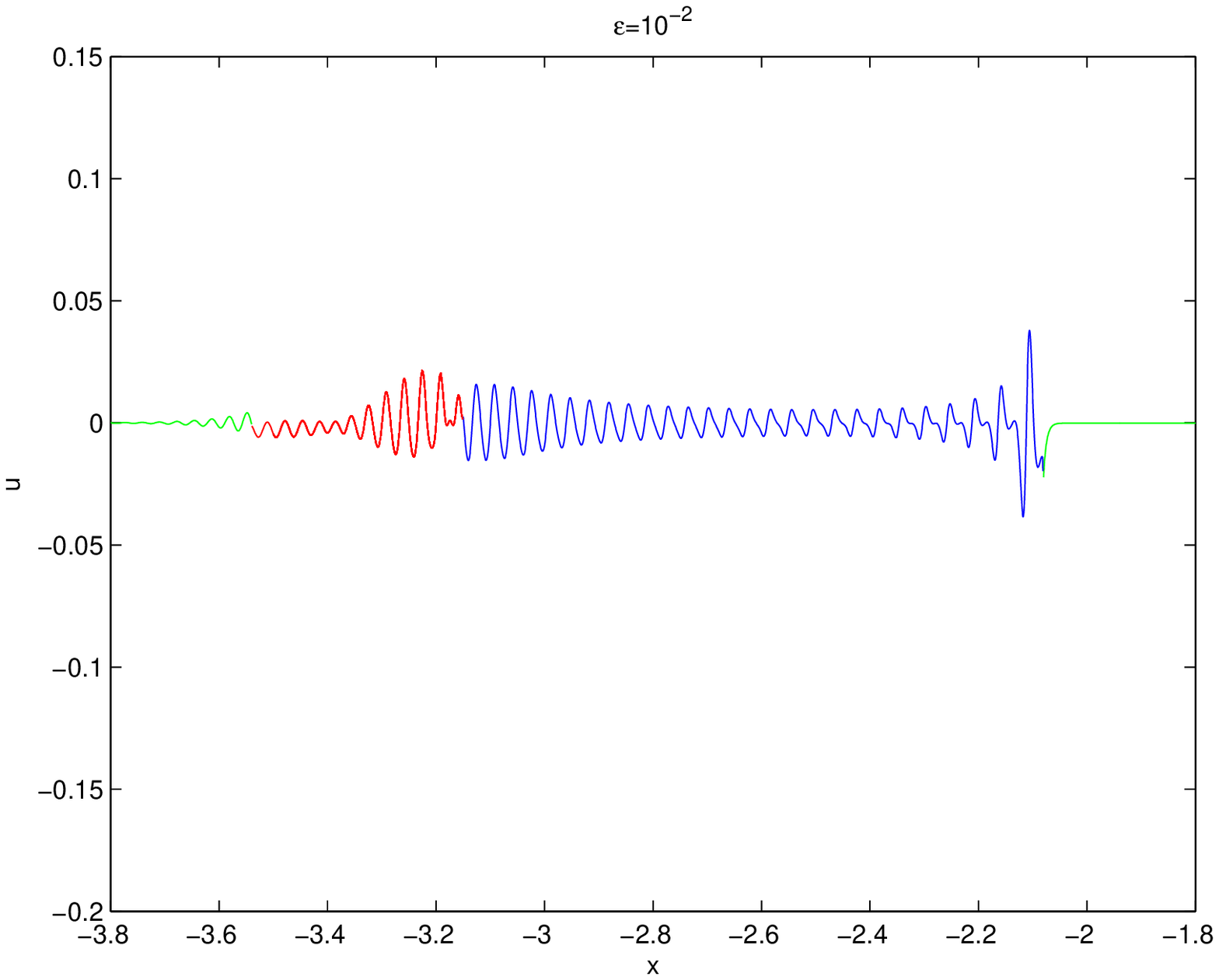, width=.55\textwidth}}}\newline
\caption{In (a) the difference between the upper two plots of 
Fig.~\ref{fig4} is shown. The Whitham zone  is 
shown in blue. In (b) one can see the same situation as in (a)  except for the region close to the 
leading edge of the Whitham zone where the difference between the KdV 
and the multiscale solution is shown in red.}
\label{fig5}
\end{figure}
In \cite{GK}  we have studied numerically the small
 dispersion limit of the KdV  equation for the concrete example of initial data of the form 
\begin{equation}
    u_{0}(x) = -\mbox{sech}^{2}x
    \label{initial}.
\end{equation}
We have compared the asymptotic description given in the works 
\cite{LL,V2,DVZ}
with the numerical  KdV solution. In \cite{GK} we have shown numerically
 that  the difference
 between the KdV solution and  the elliptic asymptotic solution at the 
center of the oscillatory zone scales like $\epsilon$ while this fails to be true at the boundary of the oscillatory zone. This fact was also observed for the Benjamin-Ono in \cite{JMS}. 
In particular at the left boundary, 
where the amplitude of the  oscillations tend to zero, the difference between the   KdV solution and  
the elliptic asymptotic solution  scales like $\epsilon^{\frac{1}{3}}$.
In this manuscript we show that the Painlev\'e-II equation describes 
the envelope of the oscillations at the leading edge where the 
oscillations tend to zero.
More precisely
\[
 u(x,t,\epsilon) = u(t)
    +\epsilon^{1/3}a\cos\left(\frac{\psi}{\epsilon}\right)+ O(\epsilon^{2/3})
\]
where $u(t)=u(x^-(t),t)$ solve the Hopf equation at the boundary $x^-(t)$ of the oscillatory zone, $\psi$ is a phase determined in (\ref{psi})
and   $a$ is, up to shift and rescalings, the Hastings-McLeod solution of  
Painlev\`e-II equation
\[
 A_{zz}=zA+2A^{3}
\]
with boundary conditions
\[
\lim_{z\rightarrow -\infty}A(z)=\sqrt{-z/2},\quad 
\lim_{z\rightarrow+\infty}A(z)=\mbox{Ai}(z)
\]
where $\mbox{Ai}(z)$ is the Airy function.

Painlev\'e equations appear in many branches of mathematics (for a review see \cite{Cl}). 
For example   in the study of 
self-similar solutions of integrable equations, in the 
study  of the Hele-Shaw flow near singularities  \cite{fokas}, 
or in double-scaling limits  in random matrix models (see e.g. \cite{BZ}, \cite{FIK},\cite{BI},\cite{CKV}). 
In this work, following \cite{KS},
 we perform  a double-scaling limit of the KdV equation
to derive the asymptotic description  of the leading  edge 
oscillations which are formed in the KdV small dispersion limit. We show that 
the envelope  of the oscillations is determined by the Hastings-McLeod  \cite{HM} 
solution of the Painlev\'e-II equation. 
Then we compare numerically for the initial data $u_0(x)=-\mbox{sech}^2x$ 
at the leading edge of the oscillatory front,
the KdV solution with the derived multiscale solution and show that 
the difference between the two solutions scales like 
$\epsilon^{\frac{2}{3}}$.  We identify a neighborhood
of the leading edge of the Whitham zone where the multiscale 
solution gives a better asymptotic description than the  asymptotic 
solution based on the elliptic and the Hopf solution. This allows to 
patch different asymptotic descriptions to provide a more 
satisfactory treatment of the small  dispersion limit of KdV as shown in 
Fig.~\ref{fig5}. \\
 Our analytical investigation of the multiscale expansion of the KdV solution
 requires the following assumptions on the initial data:
\begin{itemize}
\item $u_0(x)$ is negative with a single minimum; we chose the minimum to be at $x=0$ and normalized to -1, namely $u_0(0)=-1$;
\item the function $f_-(u)$ which is the inverse of the monotone decreasing part of the initial data $u_0(x)$ is such that $f'''_-(u)<0$;
\item $\int_{-\infty}^{+\infty}u_0(x)(1+x^2)dx<\infty$.
\end{itemize}
The last condition guarantees that   the solution of the  Cauchy problem for 
 KdV exists for all times $t>0$.

This manuscript is organized as follows. In section 2 we review 
the small dispersion limit of the KdV equation and study the small amplitude limit of the oscillations.
In section 3 we perform a multiscale expansion of the KdV 
solution when the oscillations tend to zero, and we show that the 
envelope of the oscillations is given by a  solution of the 
Painlev\'e-II equation.
In section 4 we numerically compare the KdV solution with the 
multiscale solution obtained in section 2. We show that the difference between 
the two solutions scales as $\epsilon^{\frac{2}{3}}$ which is in 
accordance with our analytical result. We identify a zone near the 
leading edge where the multiscale solution provides a better 
description than the elliptic and the Hopf solution and patch the 
solutions. In section 5 we summarize the results and add some 
concluding remarks on future directions of research.

\section{Asymptotic solution of KdV in the small dispersion limit}
The solution of the Cauchy problem
for the KdV equation is characterized by the appearance of a zone of 
fast oscillations of wave-length of order $\e$, see e.g.~Fig.~\ref{fig1}. 
These oscillations were called by Gurevich and Pitaevski
dispersive shock waves \cite{GP}.
Following the work of \cite{LL}, \cite{V2} and  \cite{DVZ},
the  description of  the small dispersion limit of the 
KdV equation is   the following:\\
1) for $0\leq t< t_c$, where $t_c$ is a critical time,  
the  solution $u(x,t,\epsilon)$ of the KdV  Cauchy problem  is approximated, 
for small $\e$, by $u(x,t)$ which solves the Hopf equation 
\begin{equation}
\label{Hopf}
u_t+6uu_x=0.
\end{equation}
Here $t_c$ is the time when the first
point  of gradient catastrophe appears in the solution 
\begin{equation}
\label{Hopfsol}
u(x,t)=u_0(\xi),\quad x=6tu_0(\xi)+\xi,
\end{equation}
of the Hopf equation. 
From the above, the time $t_c$ of gradient catastrophe can be
evaluated from the relation
\[t_{c}=\dfrac{1}{\max_{\xi\in\mathbb{R}}[-6u_0'(\xi)]}.
\]
2) After the time of gradient catastrophe, 
the solution of the KdV equation is characterized by the
  appearance  of an interval of rapid modulated oscillations.  
According to the Lax-Levermore theory, the interval $[x^-(t), x^+(t)]$ of the oscillatory zone is 
independent of $\epsilon$. Here $x^-(t)$ and $x^+(t)$  are
 determined from the initial data and satisfy the condition  $x^-(t_c)=x^+(t_c)=x_c$ where $x_c$ is the $x$-coordinate of the point of gradient catastrophe of the Hopf solution.
Outside the interval  $[x^-(t), x^+(t)]$ the leading order asymptotics of $u(x,t,\e)$  as $\e\ra 0$  is described by the solution of the Hopf equation (\ref{Hopfsol}).
Inside  the interval  $[x^-(t), x^+(t)]$ the solution
 $u(x,t,\e)$  is  approximately  
described, for small $\e$,  by  the elliptic solution of KdV 
\cite{GP}, \cite{LL}, \cite{V2}, \cite{DVZ},
\begin{equation}
\label{elliptic}
u(x,t,\e)\simeq +\beta_1+\beta_2+\beta_3+2\alpha+ 2\e^2\frac{\partial^2}{\partial
x^2}\log\theta(\Omega(x,t);\mathcal{T})
\end{equation}
where
\begin{equation}
\label{Omega}
\Omega=\dfrac{\sqrt{\beta_1-\beta_3}}{2\e K(s)}[x-2 t(\beta_1+\beta_2+\beta_3) -q]
\end{equation}
and
\begin{equation}
\label{alpha}
\alpha=-\beta_{1}+(\beta_{1}-\beta_{3})\frac{E(s)}{K(s)},\;\;\mathcal{T}=i\dfrac{K'(s)}{K(s)},
\;\; s^{2}=\frac{\beta_{2}-\beta_{3}}{\beta_{1}-\beta_{3}}
\end{equation}
with  $K(s)$ and $E(s)$ the complete elliptic integrals of the first 
and second kind, $K'(s)=K(\sqrt{1-s^{2}})$;
 $\theta$ is the Jacobi elliptic theta function defined by the 
Fourier series
\begin{equation}
\label{theta}
\theta(z;\mathcal{T})=\sum_{n\in\mathbb{Z}}e^{\pi i n^2\mathcal{T}+2\pi i nz}.
\end{equation}
For constant values of the  $\beta_i$ the formula (\ref{elliptic}) is an exact solution of KdV well
known in the theory of finite gap integration \cite{IM}, \cite{DN0}. 
However in  the description 
of the leading order asymptotics of $u(x,t,\e)$ as $\e\ra 0$,
 the quantities $\beta_i$ depend on $x$ and $t$  and evolve
 according to the Whitham equations \cite{W}
\begin{equation}
\label{Whitham}
\dfrac{\partial}{\partial t}\beta_i+v_i\dfrac{\partial}{\partial x}\beta_i=0,\quad i=1,2,3,
\end{equation}
where the speeds $v_i$ are given by the formula
\begin{equation}
    v_{i}=4\frac{\prod_{k\neq
     i}^{}(\beta_{i}-\beta_{k})}{\beta_{i}+\alpha}+2(\beta_1+\beta_{2}+\beta_{3}),
    \label{eq:la0}
\end{equation}
with $\alpha$ as in (\ref{alpha}). 
The formula for $q$ in the phase $\Omega$ in (\ref{Omega}) that we are giving below was introduced in \cite{GK} 
and  looks different but  is equivalent to the one in \cite{DVZ}
\begin{equation}
\label{q0}
    q(\beta_{1},\beta_{2},\beta_{3}) = \frac{1}{2\sqrt{2}\pi}
    \int_{-1}^{1}\int_{-1}^{1}d\mu d\nu \frac {f_-( \frac{1+\mu}{2}(\frac{1+\nu}{2}\beta_{1}
	+\frac{1-\nu}{2}\beta_{2})+\frac{1-\mu}{2}\beta_{3})}{\sqrt{1-\mu}
    \sqrt{1-\nu^{2}}},
\end{equation}
where $f_-(y)$ is the inverse function of the decreasing part of the initial data $u_0$.
The above formula for $ q(\beta_1, \beta_2, \beta_3)$ is valid as long as $\beta_1 > \beta_2 > \beta_3 > -1$. When $\beta_3$ reaches the minimum value $-1$ 
and passes over the negative hump, it is necessary to take into account 
also the increasing part of the initial data $f_+(u)$ in formula (\ref{q0}). 
We denote by $T$ this time. For $t > T>t_c$ we introduce the variable $X_3$ 
defined by $u_0(X_3) = \beta_3$ which is still monotonous. For values of $X_3$ beyond the hump, namely $X_3 > 0$, we have to substitute (\ref{q0}) 
by the formula 
\begin{equation}
\label{q01}
    q(\beta_{1},\beta_{2},\beta_{3}) = \frac{1}{\sqrt{2}\pi}
    \int_{\beta_2}^{\beta_1}d\lambda\dfrac{\left(\int_{\beta_3}^{-1}d\mu \frac {f_+(\mu)}{\sqrt{\lambda-\mu}}+\int_{-1}^{\lambda}
\frac {f_-(\mu)}{\sqrt{\lambda-\mu}}\right)}{\sqrt{(\beta_1-\lambda)(\lambda-\beta_2)(\lambda-\beta_3)}}.
\end{equation}
The  function $q=q(\beta_1,\beta_2,\beta_3)$ is symmetric with respect to $\beta_1,\beta_2$ and 
$\beta_3$,  and satisfies a linear
over-determined system of Euler-Poisson-Darboux type. It  has been  
introduced in the work of Fei-Ran Tian \cite{FRT1}.
The Whitham equations  (\ref{Whitham}) can be integrated through the   
so called hodograph transform, which   generalizes the method of characteristics, 
and which gives the solution   in the implicit form \cite{T}
\begin{equation}
\label{hodograph}
x=v_it+w_i,\quad i=1,2,3,
\end{equation}
where  the $v_i$  are defined in (\ref{eq:la0}) and the $w_i=w_i(\beta_1,\beta_2,\beta_3)$  
are obtained from an algebro-geometric procedure \cite{K} by the formula \cite{FRT1}
\begin{equation}
    w_{i} =
    \frac{1}{2}\left(v_{i}-2\sum_{k=1}^{3}\beta_{k}\right)\frac{
    \partial q}{\partial\beta_{i}}+q,\quad i=1,2,3,
    \label{eq:w}
\end{equation}
with $q$ defined in (\ref{q0}) or (\ref{q01}).
Formula (\ref{hodograph}) solves  the initial value problem for the Whitham equations  (\ref{Whitham})
 with  boundary conditions:\\
\noindent
a) {\it leading   edge:} 
\begin{equation}
\label{leadboundary}
\begin{split}
&\beta_1=\mbox{ the Hopf solution (\ref{Hopf})}\\ 
&\beta_2=\beta_3,
\end{split}
\end{equation}
b) {\it trailing  edge:}
\begin{equation}
\label{trailboundary}
\begin{split}
&\beta_2=\beta_1\\
&\beta_3=\mbox{ the Hopf solution (\ref{Hopf})}.
\end{split}
\end{equation}
In \cite{GK} we have solved numerically the initial value problem for the Whitham equations. In this way
 we could perform a numerical comparison between  the KdV small dispersion solution 
and the asymptotic formula (\ref{elliptic}) (see Fig.~\ref{fig5}).
While in the interior of the oscillatory zone the error scales 
numerically like $\epsilon$,
at the left boundary of the oscillatory zone the error scales numerically like $\epsilon^{\frac{1}{3}}$.
To derive a more satisfactory asymptotic approximation of the KdV 
small dispersion limit in the vicinity of this point, 
we perform in the next section a double scaling expansion of the KdV equation, following the double 
scaling limits appearing in random matrix theory. Before doing this analysis, we study the elliptic solution (\ref{elliptic}) in the limit when the oscillations go to zero.

\subsection{Small amplitude limit of the elliptic solution}
 We study the elliptic solution (\ref{elliptic}) near the leading edge, namely  when oscillations go to zero.
To avoid degeneracies,  we rewrite  the system (\ref{hodograph}) 
in the  equivalent form 
\begin{equation}
\label{lead0}
\left\{
\begin{aligned}
&(v_1t+w_1-x)(\alpha+\beta_1)=0\\
&v_2t+w_2-x=0\\
&\dfrac{1}{(\beta_2-\beta_3)}[(v_2-v_3)t+w_2-w_3]=0.
\end{aligned}
\right.
\end{equation}
and perform the limit $\delta\rightarrow 0$ where 
\[
\beta_2=v+\delta,\quad \beta_3=v-\delta,\quad   \beta_1=u.
\]
To simplify our calculation we restrict ourselves to the case $t_c<t<T$.
The following limit holds:
\begin{equation}
\label{elliptictrail}
\dfrac{E(s)}{K(s)}=1-\dfrac{\delta}{v-\beta_1}+\dfrac{3}{4}\dfrac{\delta^2}{(v-\beta_1)^2}+O(\delta^3)
\end{equation}
such that 
\begin{equation}
\label{alphaexp}
\alpha=-v-\dfrac{\delta^2}{4(u-v)}.
\end{equation}
Furthermore the following identities hold
\begin{align}
\label{eqex1}
f_-(u)=&[2(u-v)\partial_{u}q(u,v,v)+q(u,v,v)]\\
\label{eqex2}
\Phi(v,u)=&\partial_{v}q(u,v,v)+\partial_{u}q(u,v,v)
\end{align}
where
\begin{equation}
    \Phi(v,u)= \frac{1}{2\sqrt{2}}\int_{-1}^{1}d\mu \frac{f'_-(
\frac{1+\mu}{2}v+\frac{1-\mu}{2}u)}{
    \sqrt{1-\mu}}=
\frac{1}{2\sqrt{v-u}}\int_{u}^{v}d\mu\frac{f'_-(\mu)}{
    \sqrt{v-\mu}}.
    \label{Phi}
\end{equation}
Substituting  (\ref{alphaexp}) (\ref{eqex1}) and (\ref{eqex2}) into (\ref{lead0}) 
we arrive at the system 
\begin{equation}
\label{lead00}
\left\{
\begin{aligned}
x&= 6tu+f_-(u)-\delta^2\dfrac{(x-6tu-f_-(u)-2(u-v)(6t+\Phi(v;u))}{8(v-u)^2}+O(\delta^4),
\\
x&=6tu+f_-(u)+2(v-u)[6t+\Phi(v,u)]+\delta[6t+\Phi(v,u)+(v-u)\partial_{v}\Phi(v,u)]\\
&+\dfrac{\delta^2}{4(u-v)}[6t-2(u-v)^2\partial_{vvv}q(u,v,v)+4(u-v)\partial_{vv}q(u,v,v)+\dfrac{3}{2}\partial_v q(u,v,v,)]+O(\delta^3)\\
0&=6t+\Phi(v,u)+(v-u)\partial_{v}\Phi(v,u)+O(\delta).
\end{aligned}
\right.
\end{equation}
From the above we deduce that, in the limit $\delta\rightarrow 0$, 
the hodograph transform (\ref{lead0}) reduces to the form  (see \cite{FRT1}\cite{GT})
\begin{equation}
\label{lead}
\left\{
\begin{aligned}
&6ut+f_-(u)-x=0\\
&    \Phi(v,u)+6t =  0\\
    &\partial_{v}\Phi(v,u)  =  0.
    \end{aligned}
\right.
\end{equation}
The above system enables one to  determine $x$, $u$ and
$v$ as a function of time. This time dependence will be denoted $x=x^-(t)$, $u=u(t)$ and $v=v(t)$.
We are interested in studying the behavior of the elliptic solution
(\ref{elliptic})
 near the leading edge, namely when $x-x^-(t)$ is small and $x>x^-(t)$.
For this purpose we introduce two unknown functions of $x$ and $t$,
\[
\delta=\delta(x-x^-(t)),\quad \Delta=\Delta(x-x^-(t))
\]
which  tend to zero as $x\rightarrow x^-(t)$. We are going to 
  derive the dependence of $\Delta $ as a function of $x-x^-(t)$.
Let us fix
\begin{equation}
\label{limit}
\beta_2=v+\delta, \quad \beta_3=v-\delta,\;\;\delta\rightarrow 0
\quad \beta_1=u+\Delta,\;\;\Delta\rightarrow 0.
\end{equation}
Using the first equation of  (\ref{lead00})   we obtain 
\[
0= x-6tu-f_-(\beta_1)+\delta^2\dfrac{(x-6t\beta_1-f_-(\beta_1)-2(\beta_1-v)(6t+\Phi(v;\beta_1))}{8(\beta_3-u)^2}+O(\delta^4).
\]
Expanding the above expression near $\beta_1(x,t)=u(t)+\Delta(x,t)$,  using the identity
\begin{equation}
\label{dPhi}
\dfrac{\partial}{\partial \beta_1}\Phi(\beta_3;\beta_1)=\dfrac{\Phi(\beta_3;\beta_1)-\Phi(\beta_1;\beta_1)}{2(\beta_3-\beta_1)}
\end{equation}
and (\ref{lead}) we arrive at the expression
\begin{equation}
0\simeq x-x^-(t)-(6t+f'_-(u))\Delta+\dfrac{\delta^2}{8(v-u)^2}(x-x^-(t)),
\end{equation}
so that 
\begin{equation}
\label{Delta}
\Delta\simeq \dfrac{x-x^-(t)}{6t+f'_-(u)}.
\end{equation}
Using the second equation in (\ref{lead00}) we arrive at
\begin{equation}
\label{delta}
x-x^-(t)\simeq \delta^2\,c
\end{equation}
where 
\begin{equation}
\label{c}
\begin{split}
c&=[6t-2(u-v)^2\partial_{vvv}q(u,v,v)+4(u-v)\partial_{vv}q(u,v,v)+
\dfrac{3}{2}\partial_v q(u,v,v,)]/4(u-v)\\
&=-\dfrac{u-v}{2}\partial_{vv}\Phi(v;u).
\end{split}
\end{equation}
Therefore 
\[
\dfrac{\delta^2}{\Delta}=O(1).
\]
\begin{theorem}
The elliptic solution (\ref{elliptic}) in the limit (\ref{limit}) takes the form
\begin{equation}
\label{ellipticos}
u(x,t,\e)\simeq  u(t)+\dfrac{x-x^-(t)}{6t+f'_-(u)}+2\delta\cos\left(2\pi\dfrac{\Omega^-}{\e}\right)+
\delta^2\dfrac{\cos\left(4\pi\dfrac{\Omega^-}{\e}\right)-1}{2(u(t)-v(t))},
\end{equation}
where the phase $\Omega^-$ takes the form
\begin{equation}
\label{phasetrail}
2\pi\Omega^-=\phi_0+\phi_2
\end{equation}
with
\begin{equation}
\label{eta12}
\phi_0(t)=-16\int_{t_c}^t(u(\tau)-v(\tau))^{\frac{3}{2}}d\tau,\;\;\; \phi_2(x,t)=2\sqrt{u(t)-v(t)}(x-x^-(t)),
\end{equation}
and $u(t)$, $v(t)$ and $x^-(t)$ solve the system (\ref{lead}).
\end{theorem}
\begin{proof}
We first prove the relation (\ref{phasetrail}).
Using the expansion 
\[
K(s)=\dfrac{\pi}{2}\left(1+\dfrac{s^2}{4}+\dfrac{9}{64}s^4+O(s^6)\right),\quad
\]
and (\ref{eqex1}) we obtain the following limit for the phase $\Omega$ in 
(\ref{Omega}) 
\begin{align*}
2\pi\Omega|_{\beta_{2,3}=v\pm\delta}&=
2\sqrt{\beta_1-v}(1-\dfrac{3\delta^2}{16(\beta_1-v)^2})
[x-6t\beta_1-f_-(\beta_1)\\
&+2(\beta_1-v)(2t+\partial_{\beta_1}q(\beta_1,v,v)
-\dfrac{\delta^2}{4}\partial^2_{v}q(\beta_1,v,v)]+O(\delta^4).
\end{align*}
Using the identity
\[
\partial_{v}^2q(\beta_1,v,v)=\partial_v\Phi(v,\beta_1)-\partial_{\beta_1}\partial_v q(\beta_1,v,v),
\]
  (\ref{lead}) and (\ref{eqex2}) we can rewrite the above in the form
\[
\partial_{v}^2q(\beta_1,v,v)=\dfrac{3(2t+\partial_{\beta_1}q(\beta_1,v,v))}{2(v-\beta_1)}+\dfrac{3}{4}\dfrac{6t+f'_-(u)}{(v-u)^2}\Delta+O(\Delta^2)
\]
so that the phase  $\Omega$ takes the form
\begin{equation}
\label{om3}
2\pi\Omega|_{\beta_{2,3}=v\pm\delta}\simeq
4(\beta_1-v)^{\frac32}(2t+\partial_{\beta_1}q(\beta_1,v,v))-
\dfrac{\delta^2}{8} \dfrac{x-x^-(t)}{(v-u)^{\frac32}}.
\end{equation}
We define
\begin{align}
\nonumber
\eta_0(\beta_1,v):=&4\sqrt{\beta_1-v}[2(\beta_1-v)t+(\beta_1-v)\partial_{\beta_1}q(\beta_1,v,v)]\\
\label{phi0b}
=&2\int_{v}^{\beta_1}\sqrt{\beta_1-\lambda}[\Phi(\lambda,\beta_1)+6t]d\lambda,
\end{align}
so that, expanding near $\beta_1=u+\Delta$,  by (\ref{lead})
\begin{equation}
\begin{split}
\label{phi0c}
\eta_0(\beta_1,v)=
\eta_0(u,v)+\phi_2(x,t)+O(\Delta^2),
\end{split}
\end{equation}
where $\phi_2(x,t)$ is defined in (\ref{eta12}).
To show that $\eta_0(u,v)$ defined in (\ref{phi0b}) coincides with $\phi_0$ 
defined in (\ref{eta12}), we differentiate (\ref{phi0b}) with respect to time,
\[
2\dfrac{d}{dt}\int_{v}^{u}\sqrt{u-\lambda}[\Phi(\lambda,u)+6t]d\lambda=-16(u-v)^{\frac{3}{2}}
\]
where we have used the identity (\ref{dPhi}) and 
$\partial_t u(t)=12 \dfrac{(v-u)}{6t+f'_-(u)}.$ Integrating the r.h.s 
of the above expression with respect to $t$ 
from $t_c$ to $t$ we obtain the formula (\ref{eta12}).

Using (\ref{phi0c})  and (\ref{Delta}) we rewrite the phase (\ref{om3}) 
 in the form
\begin{equation}
\label{Omf}
2\pi\Omega|_{\substack{\beta_1=u+\Delta\\\beta_{2,3}=v\pm\delta}}
=\phi_0+\phi_2-\dfrac{\delta^2(x-x^-(t))}{8(u-v)^{\frac{3}{2}}}+O(\Delta^2)
\end{equation}
where $\phi_0$ and $\phi_2$ are as in (\ref{eta12}). Neglecting the higher order terms in $\delta$
and $\Delta$ of the above expansion one obtains (\ref{phasetrail}).

Now we are ready to expand the theta-function expression in the limit of small amplitudes.
Using (\ref{alphaexp}) and
\[
e^{i\pi \mathcal{T}}=\dfrac{\delta}{8(u-v)}(1-\frac{\Delta}{u-v})+\dfrac{5}{128}\dfrac{\delta^2}{(u-v)^3}+O(\delta^5,\delta^3\Delta),
\]
 one derives the small amplitude limit of the Jacobi $\theta$-function (\ref{theta})
\[
\theta(z;\tau)=1+\dfrac{\delta}{4(u-v)}(1-\frac{\Delta}{u-v})\cos(2\pi z)+O(\delta^4).
\]
Substituting the above expansion in  (\ref{elliptic}) one obtains 
\[
u(x,t,\e)\simeq u(t)+\dfrac{x-x^-(t)}{6t+f'_-(u)}+2\delta(x,t)\cos\left(2\pi\dfrac{\Omega^-}{\e}\right)+
\delta(x,t)^2\dfrac{\cos\left(4\pi\dfrac{\Omega^-}{\e}\right)-1}{2(u(t)-v(t))}
.
\]
which coincides with (\ref{ellipticos}).
\end{proof}

\section{Painlev\'{e} equations at the leading edge}
In this section we present a multiscale description of the 
oscillatory behavior of a solution to the KdV equation in the small
dispersion limit close to the leading edge $x^-(t)$ where 
$\beta_2=\beta_3=v$ and $\beta_1=u$. 
We are interested in the double scaling limit to the solution of 
the KdV equation (\ref{p22})
as   $\epsilon\rightarrow 0$  and $x\rightarrow x^-(t)$ in such a way that 
$x-x^-(t)\propto \epsilon^{2/3}$.
 We introduce the rescaled coordinate 
$y$ near the leading edge, 
\begin{equation}
    y=\epsilon^{-2/3}(x-x^-(t)),
    \label{p23}
\end{equation}
which transforms the KdV equation (\ref{p22}), 
to the form
\begin{equation}
   \epsilon^{\frac{2}{3}} u_{t}+\epsilon^{\frac{2}{3}} u_{yyy}+(6u-x^-_{t})u_{y}=0
    \label{p24},
\end{equation}
where $x^-_t=\dfrac{d}{dt} x^-(t)$.
The substitution (\ref{p23}) has the effect that the linear term of (\ref{p24}) 
is just the Airy equation $u_{t}+u_{yyy}=0$ which has
oscillatory solutions.

It is known \cite{DKMVZ},\cite{KS} that the corrections to the 
Hopf solution near the leading edge are of the order $\epsilon^{1/3}$.
 We thus make the  ansatz
\begin{equation}
    u(y,t,\epsilon) = U_{0}+\epsilon^{1/3}U_{1}+\epsilon^{2/3}U_{2}+
    \epsilon U_{3}+\ldots
    \label{p25},
\end{equation}
where $U_{0}=u(t)$ is the solution at the leading edge. 
We assume that $U_{k\geq 1}$ contains oscillatory 
terms with oscillations of the order $1/\epsilon$.
In particular 
\begin{equation}
    U_{1}=a(y,t)\cos \left(\frac{\psi(y,t)}{\epsilon}\right), 
    \label{p26}
\end{equation}
where 
\begin{equation}
    \psi(y,t) = 
	\psi_{0}(y,t)+\epsilon^{1/3}\psi_{1}(y,t)
	+\epsilon^{2/3}\psi_{2}(y,t)+\epsilon \psi_{3}(y,t)
+\ldots
\label{p27}.
\end{equation}
Similarly we put 
\begin{equation}
    U_{2}=b_{1}(y,t)+b_{2}(y,t)\cos\left(\frac{2\psi(y,t)}{\epsilon}
    \right),\label{p28}
\end{equation}
and
\begin{equation}
    U_{3}=c_{0}(y,t)+c_{2}(y,t)\sin\left(\frac{2\psi(y,t)}{\epsilon}\right) +c_{3}(y,t)\cos 
	\left(\frac{3\psi(y,t)}{\epsilon}\right) 
	\label{p29}.
\end{equation}
Terms proportional to  
$\sin(\psi/\epsilon)$ can be absorbed by a redefinition of $\psi$.
Since we impose no further restrictions on $\psi$ here,  such terms are therefore omitted in all orders.  We 
only consider terms proportional to $\cos(\psi/\epsilon)$ in order 
$\epsilon^{1/3}$ and the necessary terms in higher order to 
compensate the terms due to the nonlinearities in (\ref{p24}). 

If we enter equation (\ref{p22}) with this ansatz,
 we immediately obtain from the term of order $\epsilon^0$ that $\psi_{0,y}=\psi_{1,y}=0$.
From the term of  order  $\epsilon^{\frac13}$ we get
\begin{equation}
    \psi_{2,y}^{3}-(6U_{0}-x^-_{t})\psi_{2,y}-\psi_{0,t}=0
    \label{p210}.
\end{equation}
In order $\epsilon^{2/3}$ we obtain the following equations
\begin{align}
\label{p211}
    &b_{2}-\frac{a^{2}}{2\psi_{2,y}^{2}}=0\\
    \label{p212}
    &\psi_{3,y}(3\psi_{2,y}^{2}-6U_{0}+x^-_{t})-\psi_{1,t}=0\\
     \label{p213}
   & \dfrac{d}{dy}[a^2(3\psi_{2,y}^{2}-6U_{0}+x^-_{t})]=0.
\end{align}
In order $\epsilon$ we get 
\begin{align}
 \label{p216}
    &U_{0,t}+(6U_{0}-x^-_t)b_{1,y}+3aa_{y}=0\\
 \label{p217}    
& \psi_{2,y}\dfrac{d}{dy}\left(a^2\psi_{3,y}\right)=0\\
   \label{p219}
    &2a^{2}\psi_{3,y}=0\\
     \label{p221}
&\psi_{2,y}^{2}a_{yy}+a(2b_1\psi_{2,y}^2+\frac{1}{3}\psi_{2,y}\psi_{2,t})+\frac{1}{2}a^3=0\\
\label{p227a} 
&c_{2}=-\frac{aa_{y}}{\psi_{2,y}^{3}}\\
&c_{3}=\frac{3a^{3}}{16\psi_{2,y}^{4}}.
   \label{p227b} 
\end{align}
A  solution to (\ref{p212}), (\ref{p213}) and (\ref{p219})  is
\begin{equation}
   3\psi_{2,y}^{2}=6U_{0}-x^-_t+\dfrac{C(t)}{a^2}, \quad \psi_{1,t}=0,
    \label{p214}
\end{equation}
where $C(t)$ is a free  function of $t$.
Note that
\begin{equation}
    U_0(t)=u(t),\quad x^-_{t}=12v(t)-6u(t)
    \label{p222},
\end{equation}
where $u(t)$ and $v(t)$ are defined in (\ref{lead}). The above equations imply that 
\begin{equation*}
    \psi_{2}(y,t)=2y\sqrt{(u-v)}+\int\dfrac{C(t)}{a^2(y)}dy+h(t),
\end{equation*}
where $h(t)$ is free  function of $t$.
Comparing the above relation with the formula of $\phi_2$ in (\ref{eta12})
 of the phase in the small amplitude expansion   we can conclude that 
$ C(t)=0$ and $h(t)=0$
so that  
\begin{equation}
    \psi_{2}(y,t)=2y\sqrt{u(t)-v(t)}
    \label{p223}.
\end{equation}
From (\ref{p210}), (\ref{p222}) and (\ref{p223}) we derive that
\[
\psi_{0,t}=-16(u-v)^{\frac32},
\]
namely
\begin{equation}
\psi_{0}(t)=-16\int_{t_c}^t(u(\tau)-v(\tau))^{\frac32}d\tau.
\end{equation} 

From  (\ref{p216}) we find
\begin{equation}
\label{p220}
\begin{split}
    b_{1}&=-\frac{a^{2}}{2\psi_{2,y}^{2}}-\frac{yU_{0,t}}{3\psi_{2,y
    }^{2}}+k(t)\\
&=-\frac{a^{2}}{8[u(t)-v(t)]}+\frac{y}{6t+f'_-(u)}+k(t),
\end{split}
\end{equation}
where $k(t)$ is a free function of $t$. It will be fixed by matching it 
with the elliptic solution in the Whitham zone. 
Substituting (\ref{p220}) and (\ref{p223}) into (\ref{p221}) we 
arrive at the equation
\begin{equation}
    4(u(t)-v(t))a_{yy}-\frac{2}{3}v_t(t)a\left(
    y-\frac{12k(t)(u(t)-v(t))}{v_t(t)}\right)=\frac{a^{3}}{2}
    \label{p224}.
\end{equation}
Making the substitution
\begin{equation}
    A=6^{\frac{1}{3}}a/(4v_{t}^{1/3}(u-v)^{1/6}),\quad 
    z=\left(\frac{v_{t}}{6(u(t)-v(t))}\right)^{1/3}(y-y_{0})
    \label{p226}
\end{equation}
with     $y_{0}=12k(t)(u(t)-v(t))/v_{t}$, we arrive at the equation
\begin{equation}
    A_{zz}=zA+2A^{3}
    \label{p225},
\end{equation}
which is  a special case of the  Painlev\'e-II equation $  A_{zz}=zA+2A^{3}-\gamma$, with  $\gamma=0$.

Since we are only interested in terms up to order $\epsilon^{1/3}$ in 
$u(x,t,\epsilon)$, the terms $b_{1}$, $b_{2}$, $c_{0}$, $c_{2}$ and $c_{3}$ are not important for us. 
However, we had to go to order $\epsilon$ to determine $\psi_{3}$ 
which will contribute to the $\epsilon^{1/3}$ terms in $u$. 

To sum up we get for $u(x,t,\e)$
\begin{equation}
 \label{p228.1}
    u(x,t,\epsilon) = u(t)
    +\epsilon^{1/3}a\cos\left(\frac{\psi}{\epsilon}\right)+\epsilon^{2/3}\left[ \dfrac{a^2(\cos(2\psi/\e)-1)}
{8(u(t)-v(t))}+k(t)+\dfrac{y}{6t+f'_-(u)}\right]+
O(\epsilon)
\end{equation}
where
\[
\psi(y,t)=-16\int_{t_c}^t(u(\tau)-v(\tau))^{\frac32}d\tau+\epsilon^{\frac{1}{3}}h+
\epsilon^{\frac{2}{3}}[2y\sqrt{u(t)-v(t)}]+
\epsilon\psi_3(t)+O(\e^{\frac{4}{3}}),
\]
and $h$ is an integration constant. There are free functions of $t$ in the integration of the 
multi-scale equations, namely the functions  $k(t)$ and $\psi_3(t)$ and the constant $h$.
Moreover, the solution of the Painlev\'e II equations needs to be fixed.
We  fix the constants by comparing  the multiscale solution (\ref{p228})  
 with the elliptic solution (\ref{ellipticos})  at the border of the Whitham zone.
Indeed comparing (\ref{p228.1}),  (\ref{ellipticos}) and (\ref{Omf}) we obtain 
\begin{equation}
\label{delta1}
\delta=\frac{1}{2}\e^{\frac{1}{3}}a,\quad \Delta=O(\e^{\frac{2}{3}})
\end{equation}
and $
 k(t)=0,\;\; \psi_3(t)=0$ and $h=0$.
Therefore the multiscale solution takes the form
\begin{equation}
 \label{p228}
    u(x,t,\epsilon) = u(t)
    +\epsilon^{1/3}a(y,t)\cos\left(\frac{\psi}{\epsilon}\right)+\epsilon^{2/3}\left[ \dfrac{a^2(\cos(2\psi/\e)-1)}
{8(u(t)-v(t))}+\dfrac{y}{6t+f'_-(u)}\right]+
O(\epsilon),
\end{equation}
where $y=\epsilon^{-2/3}(x-x^-(t))$, $a(y,t)$ satisfies (\ref{p224}) and 
\begin{equation}
\label{psi}
\psi(y,t)=-16\int_{t_c}^t(u(\tau)-v(\tau))^{\frac32}d\tau+
2\epsilon^{\frac{2}{3}}y\sqrt{u(t)-v(t)}+O(\e^{\frac{4}{3}}).
\end{equation}
For the numerical comparison in the following section, we consider 
terms up to order $\epsilon^{1/3}$ in $u$ in (\ref{p228})
\begin{equation}
\label{p228b}
u(x,t,\epsilon) = u(t)
    +\epsilon^{1/3}a(y,t)\cos\left(\frac{\psi(y,t)}{\epsilon}\right)+O(\epsilon^{2/3}),
\end{equation}
 where $\psi(y,t)$ is as given above.
For fixing the particular solution of the Painlev\'e-II equation (\ref{p225}) 
the following considerations are needed.
For large $x<x^-(t)$, the solution of KdV is essentially approximated by 
the solution of the Hopf equation, and the term of order $\epsilon^{1/3} $ has to be negligible in (\ref{p228b}), namely $a(y)\simeq 0$ for large negative
 $y=(x-x^-(t))\epsilon^{-2/3}$.
For $x<x^-(t)$ 
\[
 z=\left(\frac{v_{t}}{6(u(t)-v(t))}\right)^{1/3}(x-x^-(t))\epsilon^{-2/3}>0
\]
because  $v_t=\dfrac{6}{(u-v)\partial_{vv}\Phi(v;u)}<0$ since   $\partial_{vv}\Phi(v;u)<0$ and $u>v$. Therefore for  $x\ll x^-(t)$ we conclude that  
\begin{equation}
\label{HM0}
\lim_{z\rightarrow +\infty} A(z)=0.
\end{equation}
For $x>x^-(t)$, from the small amplitude limit of the elliptic solution 
of the KdV equation we obtain combining (\ref{delta}), (\ref{c}) and 
(\ref{delta1})
\[
 \sqrt{-\dfrac{2(x-x^-(t))}{(u-v)\partial_{vv}\Phi(v;u)}}\simeq \delta= \dfrac{1}{2}\epsilon^{1/3} a
\]
which, in the limit $\epsilon\rightarrow 0$ or $y\rightarrow +\infty$ gives
\[
\lim_{y\rightarrow +\infty} a(y)= 2\sqrt{-\dfrac{2y}{(u-v)\partial_{vv}\Phi(v;u)}}
\]
or  equivalently, by (\ref{p226}),
\begin{equation}
\label{HM1}
\lim_{z\rightarrow -\infty}A(z)=\sqrt{-z/2}.
\end{equation}
The existence and uniqueness of the solution of (\ref{p225}) satisfying (\ref{HM0}) and  (\ref{HM1}) was first established by Hastings and Mcleod \cite{HM} 
(see also later works of \cite{Kap} and \cite{Cl}).
It is also worth noticing that the asymptotics of $A(z)$ at $z\rightarrow +\infty$ can be specified as
\begin{equation}
\label{HM2}
\lim_{z\rightarrow+\infty}A(z)=\mbox{Ai}(z)
\end{equation}
where $\mbox{Ai}(z)$ is the Airy function. Moreover the asymptotic condition 
(\ref{HM2}) characterizes the solution  $A(z)$ uniquely, so that (\ref{HM1}) and (\ref{HM2}) constitutes an example of the so called connection formula 
for the Painlev\'e equations (see e.g. \cite{FI}, \cite{CM}).

\section{Comparison of the multiscale expansion  and the asymptotic 
solution to the small dispersion KdV}

The numerical evaluation of the asymptotic solution based on the Hopf 
and the elliptic solution is described in \cite{GK}.
 To evaluate the multiscale solution (\ref{p228}), one needs in 
addition to the quantities computed there 
the Hastings-McLeod solution to the Painlev\'e-II equation. This 
solution was calculated numerically by Tracy and Widom \cite{trwi} 
with standard solvers for ordinary  differential equation and by Pr\"ahofer and Spohn 
\cite{prsp,prspweb} with in principle arbitrary precision with a 
Taylor series approach. The general family of solutions of  Painlev\'e-II such that 
$\lim_{z\rightarrow +\infty}A(z)=a\mbox{Ai}(z)$ with $a$ positive constant was studied numerically
 in \cite{Ros} and analytically in \cite{CM}.
Solutions to Painlev\'e-II in the complex 
plane were studied analytically and numerically by Fokas and Tanveer in \cite{fokas}. 
We use here an approach based on spectral 
methods which is described briefly in the appendix. This approach is 
both efficient and of high precision and can directly be combined with 
the numerics of \cite{GK}.

\paragraph{\it Times $t\gg t_{c}$}
Close to breakup the multiscale expansion is expected to be 
inefficient since it is best near the leading edge, and since at 
breakup both the leading and the trailing edge coincide. We will 
discuss this solution close to breakup below, but first we will study 
it for time $t=0.4\gg t_{c}=0.216\ldots$. 
In Fig.~\ref{figp22in11e4.413} one can see that the multiscale solution 
gives an excellent approximation of the KdV solution for $x<x^{-}(0.4)=-3.2297$ 
and in the Whitham zone close to $x^{-}$. For larger values of $x$, 
the solutions are out of phase and the values of the multiscale 
solution are shifted towards positive values.
\begin{figure}[!htb]
\centering
\epsfig{figure=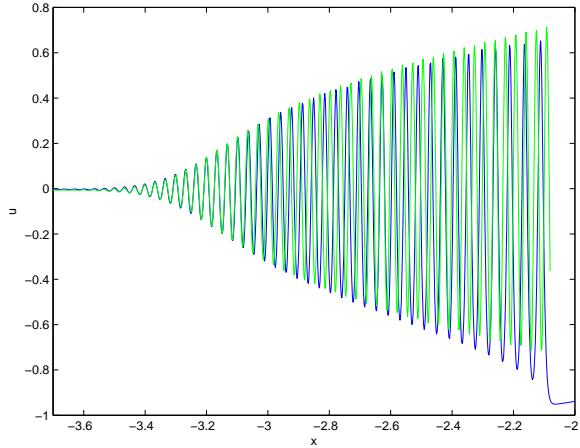, width=0.7\textwidth}
\caption{The blue line is the solution of the KdV equation for the 
initial data $u_0(x)=-\mbox{sech}^2x$ and $\epsilon=10^{-2}$ for $t=0.4$, 
and the green line is the corresponding  multiscale solution
given by formula  (\ref{p228}).}
\label{figp22in11e4.413}
\end{figure}
The difference of the two solutions is shown in 
Fig.~\ref{figp2delta1e4.413}.
\begin{figure}[!htb]
\centering
\epsfig{figure=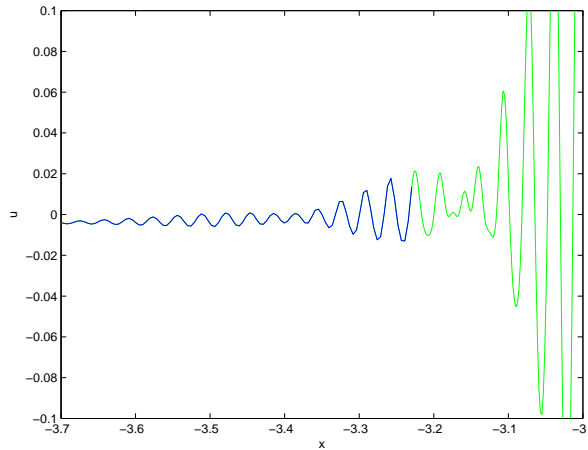, width=0.7\textwidth}
\caption{The difference of the KdV and the multiscale solution 
for the initial data $u_0(x)=\mbox{sech}^2x$ and $\epsilon=10^{-2}$ 
for $t=0.4$. The curve is plotted in green in the Whitham zone.} 
\label{figp2delta1e4.413}
\end{figure}
From this figure it is even more obvious that the multiscale solution 
is a valid approximation in the Whitham zone near the leading edge, 
but the difference increases rapidly for $|x|\gg x^{-}$. 

\paragraph{\it $\epsilon$ dependence}
In \cite{GK} it was shown that the asymptotic description becomes more accurate 
with decreasing $\epsilon$. The same is true for the multiscale 
solution as can be seen in Fig.~\ref{figp2delta4e13}. 
The zone, where the multiscale solution gives a  better
approximation than the asymptotic elliptic solution, shrinks with $\epsilon$. For $x\gg 
x^{-}(t)$, the multiscale solution is always only a poor 
approximation to the KdV solution. 
\begin{figure}[!htb]
\centering
\epsfig{figure=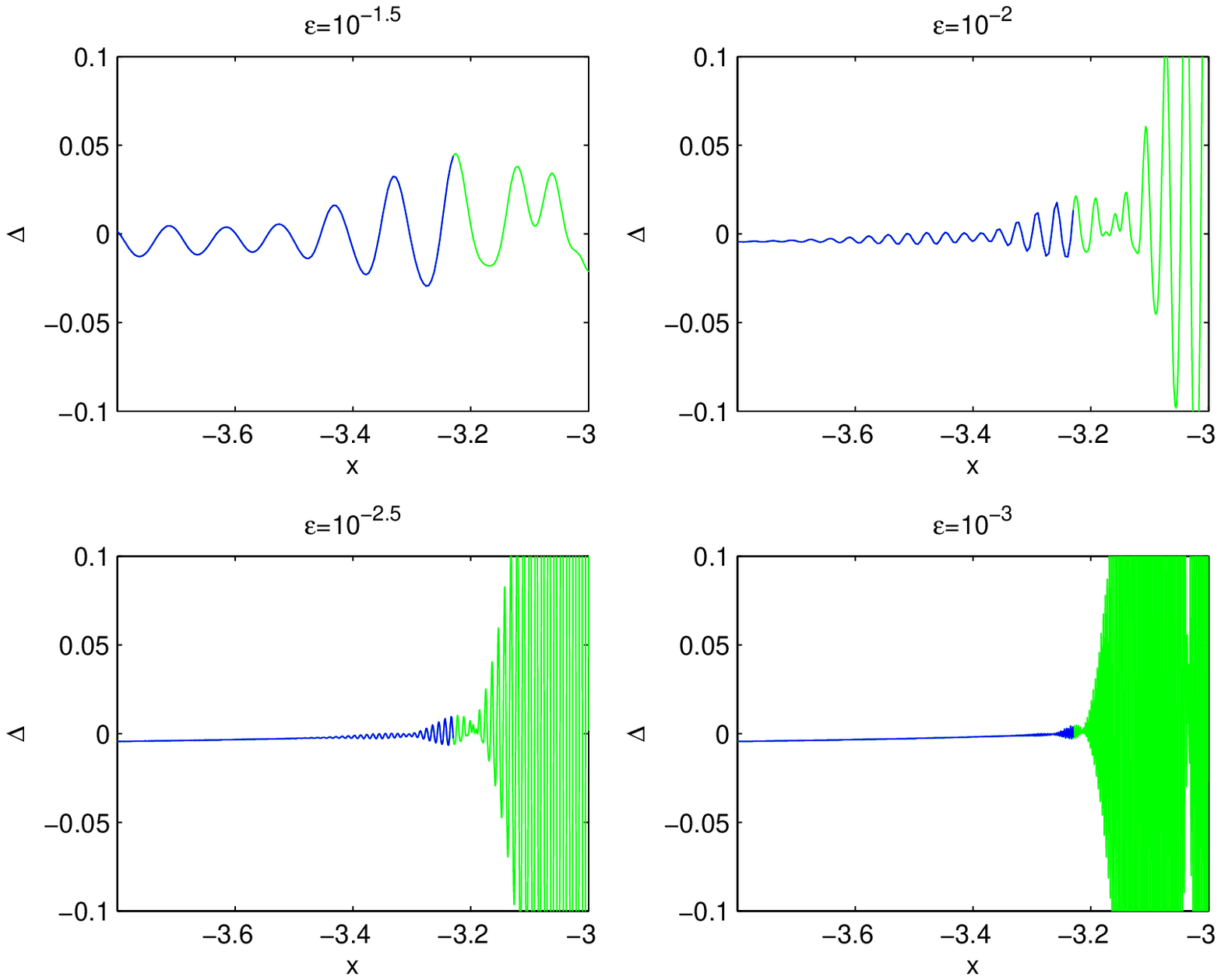, width=1.1\textwidth}
\caption{Difference of the KdV and the multiscale solution in order 
$\epsilon^{1/3}$
for the initial data $u_0(x)=-\mbox{sech}^2x$ and several values of $\epsilon$ 
for $t=0.4$. The curves are plotted in green in the Whitham zone. }
\label{figp2delta4e13}
\end{figure}
The 
maximal difference $\Delta_{max}$ of the KdV solution and the multiscale 
solution near this edge 
decreases roughly as $\epsilon^{2/3}$. More precisely the 
error can be fitted  with a straight line by a standard linear 
regression analysis, 
$-\log_{10}\Delta_{max}=-a\log_{10}\epsilon+b$ with $a=0.63$, 
$b=0.41$. The correlation coefficient is $r=0.999$, the standard 
error is $\sigma_{a}=0.02$.

\paragraph{\it Comparison and matching with the asymptotic solution}
The aim of this paper is to improve the asymptotic description of the 
small dispersion limit of KdV near the leading edge. In 
Fig.~\ref{figp2deltawhith1e413} it can be seen that the multiscale 
solution will indeed be a much better approximation near this 
edge.
\begin{figure}[!htb]
\centering
\epsfig{figure=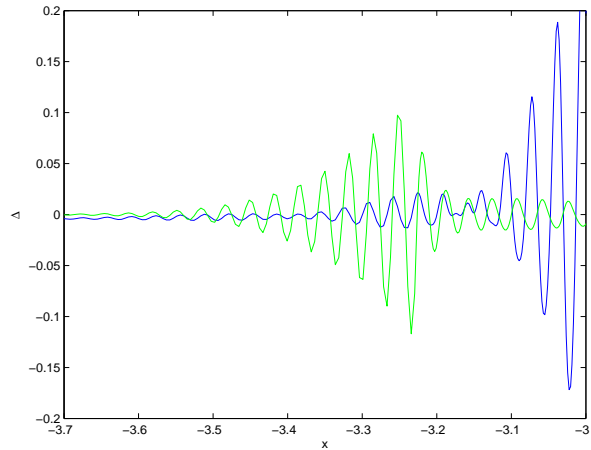, width=0.7\textwidth}
\caption{The difference of the KdV and the multiscale solution 
 (blue) and the difference  of the KdV and the asymptotic solution (green)
for the initial data $u_0(x)=-\mbox{sech}^2x$ and $\epsilon=10^{-2}$ 
for $t=0.4$. }
\label{figp2deltawhith1e413}
\end{figure}
Near the leading edge, the multiscale solution provides a superior 
description of the KdV solution, whereas the elliptic asymptotic solution is 
much better for $x\gg x^{-}(t)$ in the Whitham zone. In fact it is 
possible to identify a zone where the multiscale solution is 
more satisfactory than the asymptotic solution. Due to the strong 
oscillations of the solutions, there is a certain ambiguity in the 
definition of this zone. We define the limits of the zone as the last 
intersection (or where the solutions come closest) on which  the 
other solution has an error with larger oscillations. 
In  this  zone it is 
possible to replace the asymptotic solution by 
the multiscale solution. The result of this patch work approach is 
shown in Fig.~\ref{figp22corr2e13}. 
\begin{figure}[!htb]
\centering
\epsfig{figure=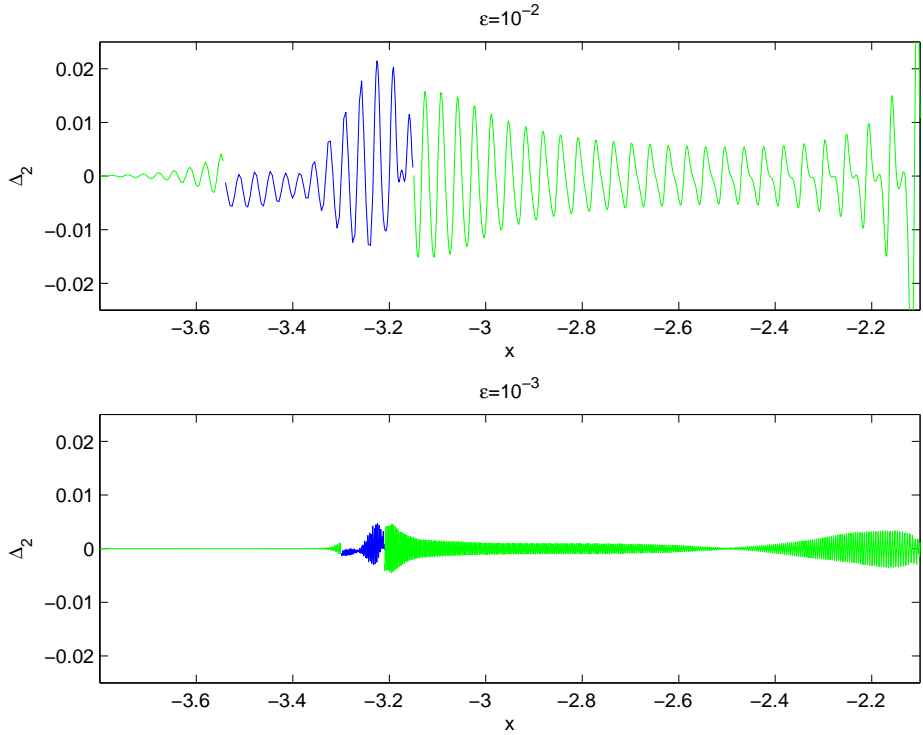, width=1.1\textwidth}
\caption{Difference of the KdV and the multiscale solution (blue) and 
the KdV and the asymptotic solution (green)
for the initial data $u_0(x)=-\mbox{sech}^2x$ 
at $t=0.4$ for two values of $\epsilon$. }
\label{figp22corr2e13}
\end{figure}
It can be seen that the resulting amended asymptotic description has 
an accuracy near the leading edge of the same order as in the 
interior of the Whitham zone. The maximal difference between the KdV 
and the asymptotic solution still occurs near the leading edge.

As already mentioned, the zone where the multiscale solution provides 
a better approximation to the KdV solution than the asymptotic 
solution, shrinks with $\epsilon$ as can be inferred from 
Fig.~\ref{figp2boulr13}. 
\begin{figure}[!htb]
\centering
\epsfig{figure=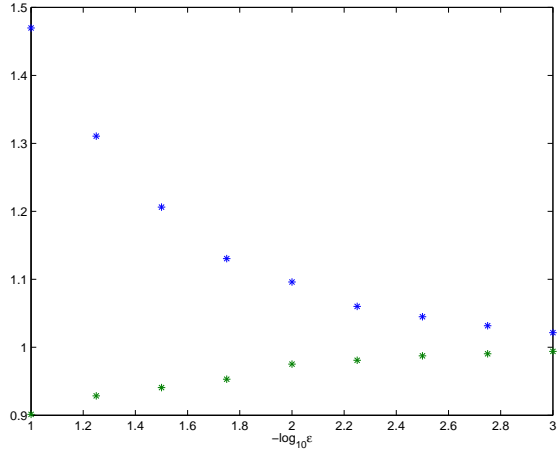, width=0.7\textwidth}
\caption{Boundary values of the zone where the multiscale solution 
provides a better approximation to the KdV solution than the 
asymptotic solution. The $x$-values of the boundary of this zone 
(normalized by $x^{-}$) for $t=0.4$ are shown for several values of $\epsilon$.}
\label{figp2boulr13}
\end{figure}
The width of this zone decreases roughly as $\epsilon^{2/3}$  which 
shows the self consistency of the used rescaling of the spatial 
coordinate near the leading edge. More
precisely, we find a scaling $\epsilon^{a}$ with $a=0.66$,  
correlation coefficient $r=0.9996$ and standard error 
$\sigma_{a}=0.015$.
It can be seen that the zone is not symmetric around the leading 
edge, it extends much further into the Hopf region than in the 
Whitham zone. This is due to the fact that the multiscale solution is 
quickly out of phase with the rapid oscillations in the Whitham zone, 
and that the Hopf solution does not have oscillations.

\paragraph{\it Breakup time}
In \cite{GK} it was shown that the elliptic  asymptotic solution is worst near the 
breakup of the Hopf solution. 
The multiscale expansion obtained in the previous section
 is not defined for times before $t_{c}$, and 
it will be worst there, since it can be understood as an expansion 
around the leading edge of the Whitham zone. At breakup, however, 
leading and trailing edge coincide. Thus the approximation is rather 
crude there, but it increases in quality with time as can be seen in 
Fig.~\ref{figp2break913}. It is, however, interesting to study at 
which times the multiscale solution starts to give a better 
asymptotic description than other approaches. Dubrovin 
conjectured \cite{dubcr} that the asymptotic behavior of the KdV solution 
close to the breakup of the corresponding Hopf solution is given by a 
particular solution to the second equation in the Painlev\'e-I 
hierarchy. In \cite{pain12} we provided strong numerical evidence for 
the validity of this conjecture. The natural question is whether 
Painlev\'e-I2 description near the critical point provides a 
satisfactory asymptotic solution for KdV till
times where the  the multiscale solution studied in the present paper 
provides a valid description near the leading edge. A comparison of 
Fig.~\ref{figp2break913} with a similar figure in \cite{pain12} shows 
that this is qualitatively the case.
\begin{figure}[!htb]
\centering
\epsfig{figure=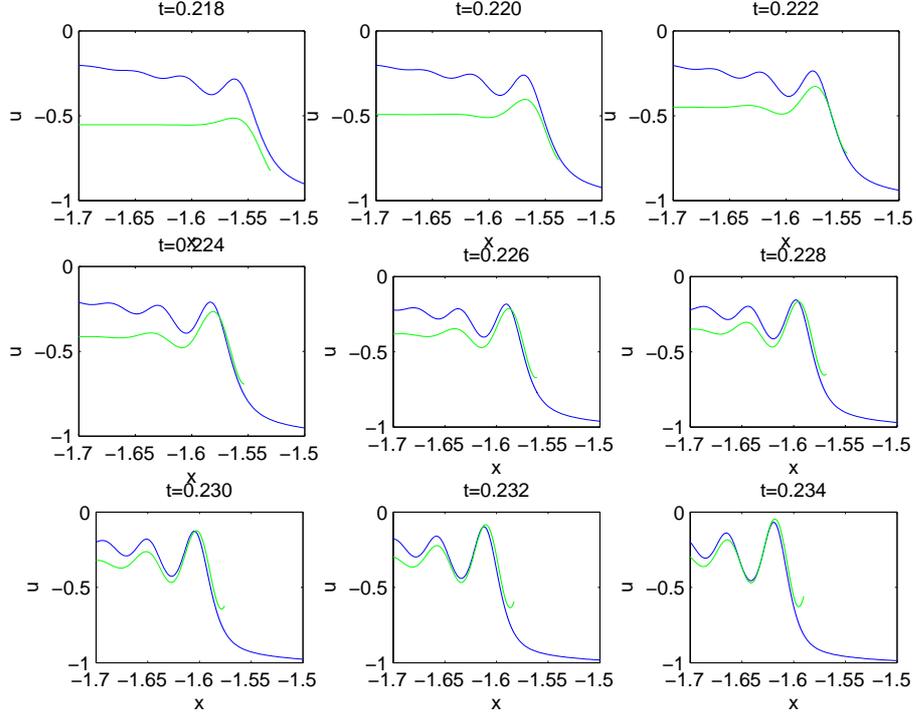, width=1.1\textwidth}
\caption{The blue line is the solution of the KdV equation for the 
initial data $u_0(x)=-\mbox{sech}^2x$ and $\epsilon=10^{-2}$, 
and the green line is the corresponding  multiscale solution 
given by formula  (\ref{p228}).
The plots are given for   different  times near the point of gradient catastrophe 
$(x_c,t_c)$ of  the Hopf solution. Here  $x_c\simeq -1.524 $, $t_c\simeq 0.216$.}
\label{figp2break913}
\end{figure}
In Fig.~\ref{figp2break913} the multiscale solution is shown for $x<x^{+}(t)$. 
Near breakup the  approximation is only acceptable close to the 
breakup point. For larger  times, more and more oscillations are 
satisfactorily reproduced by the multiscale solution.  But the multiscale solution 
will only be a better approximation of the oscillations than the Hopf 
solution for times $t\gg t_{c}$. 

\section{Outlook}
In the present work we have considered a multiscale solution to the 
KdV equation in the small dispersion limit close to the leading edge of 
the oscillatory zone. We studied the solution up to order 
$\epsilon^{1/3}$ and show that the envelope of the oscillations is described by the Hasting-McLeod solution of the Painlev\'e-II.  The validity of the approach in the considered  limit was shown numerically. 
The double scaling expansion of the KdV solution in the small dispersion
limit will be investigated with  the Riemann-Hilbert approach and steepest descent method for 
oscillatory Riemann-Hilbert problem as done in \cite{DVZ}. The Riemann-Hilbert approach seems so far the 
only analytical tool to study the double-scaling expansion to the {\it Cauchy problem} of the KdV equation. 
This project will be the subject of our future research.

As can be seen from Figure~\ref{fig5}, the asymptotic solution of the KdV equation does
not give a satisfactory description of the KdV small dispersion limit also at the trailing edge of the oscillatory zone.  This problem will be investigated in a subsequent publication, too. A similar problem
was tackled in the context of matrix-models in \cite{Claeys}.

\appendix
\section{Numerical solution of the Painlev\'e-II equation}
We are interested in the numerical computation of the Hastings-McLeod 
solution to the Painlev\'e-II 
equation 
\begin{equation}
    P_{II}A:= A_{zz}-zA-2A^{3}=0
    \label{pain1}
\end{equation}
which is subject to the asymptotic conditions \cite{HM}
\begin{equation}
    A\simeq \sqrt{-z/2} \mbox{ for } z\rightarrow -\infty,
    \label{pain2}
\end{equation}
and 
\begin{equation}
    A\simeq \mbox{Ai}(z) \mbox{ for } z\rightarrow \infty,
    \label{pain3}
\end{equation}
where $\mbox{Ai}(z)$ is the Airy function. Numerically we will consider 
equation (\ref{pain1}) on a finite interval $[z_{l},z_{r}]$ 
(typically $[-10,10]$). The asymptotic solution near $\pm \infty$, which will be discussed in 
more detail below, is truncated in a way that the truncation error at 
$z_{l}$, $z_{r}$ is below $10^{-10}$. At these points we 
impose the values following from the asymptotic solutions 
as boundary conditions, namely 
\begin{align}
     A(z_{l})& = \sqrt{-z_{l}/2}-\frac{1}{8\sqrt{2}}(-z_{l})^{-5/2} 
     -\frac{73}{128\sqrt{2}}(-z_{l})^{-11/2}
    \nonumber  \\
    A(z_{r}) & =\frac{1}{2\sqrt{\pi}z_{r}^{1/4}}\exp\left(-\frac{2}{3}
    z_{r}^{3/2}\right) 
    \label{pain4}.
\end{align}

To solve equation (\ref{pain1}) for $z\in [z_{l},z_{r}]$ we use 
spectral methods since they allow for an efficient numerical 
approximation of high accuracy. We map the interval $[z_{l},z_{r}]$ 
with a linear transformation $z\to x$ to the interval $I=[-1,1]$ and expand $A$ 
there in Chebychev polynomials. 

Let us briefly summarize the Chebychev approach, for details see 
e.g.~\cite{canuto,fornberg,cam}. The Chebyshev
polynomials $T_n(x)$ are defined on the interval $I$ by the
relation 
\[
T_n(\cos(t)) = \cos(n t)\;, \mbox{where } x = \cos(t)\;,
\qquad t\in[0,\pi]\;.
\]
A  function $f$ on $I$ is 
approximated via Chebychev polynomials, $f\approx 
\sum_{n=0}^{N}a_{n}T_{n}(x)$ where the spectral coefficients $a_{n}$
are obtained by the conditions $f(x_{l})=\sum_{n=0}^{N}a_{n}T_{n}(x_{l})$, 
$l=0,\ldots,N$. This approach is called a collocation method. If the 
collocation points are chosen to be 
$x_l=\cos(\pi l/N)$, the spectral coefficients follow from $f$ via  a 
Discrete Cosine Transform (DCT) for which fast algorithms exist. We use 
here a DCT within Matlab. A recursive relation for the derivative of 
Chebychev polynomials implies that the action of the 
differential operator $\partial_{x}$ on $f(x)$ leads to an action of a 
matrix $D$ on the vector of the spectral coefficients $a_{n}$. Thus 
we express $A(x)$ in terms of Chebychev polynomials, 
$A(x)=\sum_{n=0}^{N}\tilde{A}_{n}T_{n}(x)$ (we typically work with 
$N=128$), and the coefficients of  $\partial_{x}A$ in 
terms of Chebychev polynomials are determined then via $D\tilde{A}$. 

To solve equation (\ref{pain1}) on the interval $[z_{l},z_{r}]$, we 
use an iterative approach,
\begin{equation}
    A_{n+1,zz}=zA_{n}+2A_{n}^{3},\quad n\in \mathbb{N}
    \label{pain5}.
\end{equation}
We start with $A_{1}(z)=(1+z^{2})^{1/4}/(1+\exp(z))/\sqrt{2}$. In each step of 
the iteration we solve equation (\ref{pain5}) for $A_{n+1}$ with the 
boundary conditions (\ref{pain4}). The boundary 
conditions are imposed with a $\tau$-method: the last two rows of the 
matrix $D^{2}$ for the second derivative are replaced with the boundary conditions at 
$x=\pm1$. Since $T_{n}(\pm1)=(\pm1)^{n}$, the resulting matrix $L$ 
which will be inverted in each step of the iteration, has only $1$ 
and $-1$ in the last two rows and is thus better 
conditioned than the matrix $D^{2}$. 
It turns out that the iteration is unstable if no relaxation is used. 
We thus define $A_{n+1}=\mu L^{-1}(zA_{n}+2A_{n}^{3}) + (1-\mu) A_{n}$ 
with $\mu=0.009$. With this choice of the parameters, the iteration 
converges. It is stopped when the difference between $A_{n+1}$ and 
$A_{n}$ is of the order of machine precision (Matlab works internally 
with a precision of the order of $10^{-16}$; due to rounding errors 
machine precision is typically limited to the order of $10^{-14}$). 
The solution is shown in Fig.~\ref{figpain2}.
\begin{figure}[!htb]
\centering
\epsfig{figure=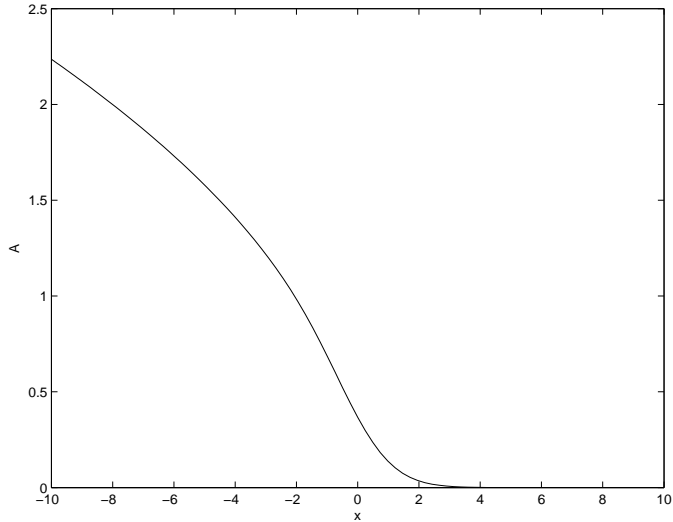, width=0.7\textwidth}
\caption{Hastings-McLeod solution of the Painlev\'e II equation.}
\label{figpain2}
\end{figure}

To test the accuracy of the solution we plot in Fig.~\ref{figp2test} 
the quantity $P_{II}A$ as computed with spectral methods on the 
collocation points. It can be seen that the error is biggest on the 
boundary which is even more obvious from Fig.~\ref{figp2testbou}. The 
found solution is also compared to a numerical solution with a 
standard ode solver as \emph{bvp4c} in Matlab. The solutions agree 
within the limits of numerical precision.
\begin{figure}[!htb]
\centering
\epsfig{figure=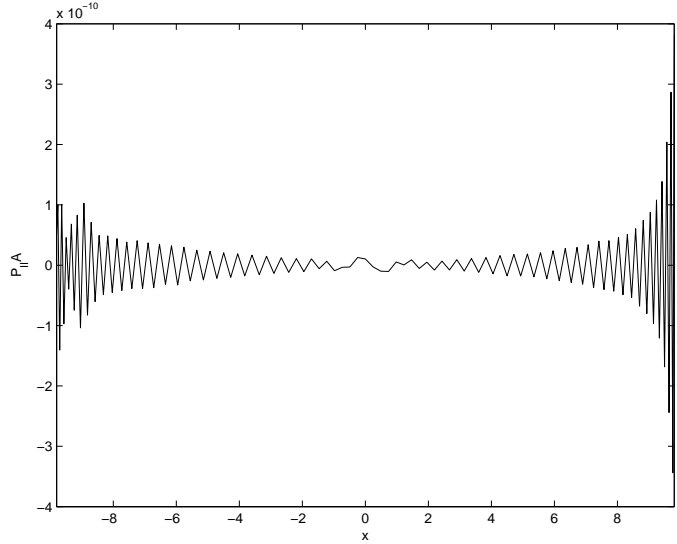, width=0.7\textwidth}
\caption{Accuracy of the solution of the Painlev\'e II equation.}
\label{figp2test}
\end{figure}
\begin{figure}[!htb]
\centering
\epsfig{figure=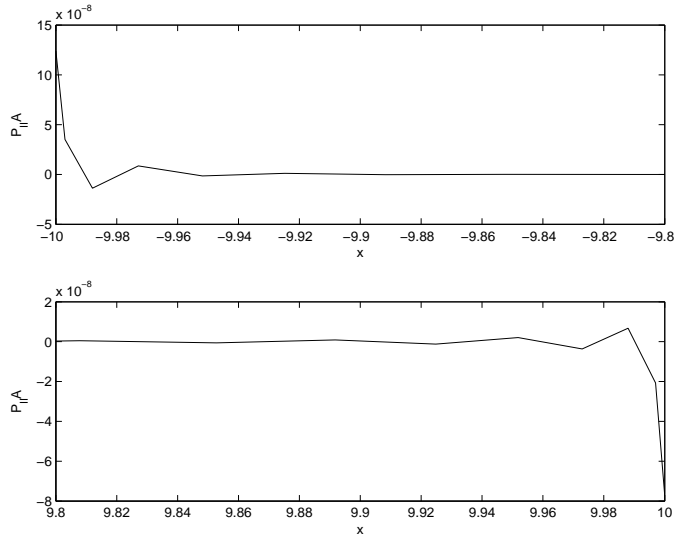, width=0.7\textwidth}
\caption{Accuracy of the solution of the Painlev\'e II equation near 
the boundary of the considered interval.}
\label{figp2testbou}
\end{figure}

For general values of $z$ the solution is obtained as follows: for 
values of $z\in [z_{l},z_{r}]$ they follow from the spectral data via 
$A(z)=\sum_{n=0}^{N}\tilde{A}_{n}T_{n}(z)$. Notice that the accuracy 
of the solution is best on the collocation points, but we can expect 
it to be of the order of at least $10^{-6}$ even at points $z$ in between. 
For values of $z<z_{l}$, we use 
the approximation $A(z)=\sqrt{-z/2}-(-z)^{-5/2}/8/\sqrt{2} -\frac{73}{128\sqrt{2}}(-z_{l})^{-11/2}$, 
for values of 
$z>z_{r}$, we use the approximation $A(z)=\exp\left(-\frac{2}{3}
z^{3/2}\right)/(2\sqrt{\pi}z^{1/4})$. 
This provides a global approximation to the 
solution with an accuracy of the order of $10^{-6}$ and better, which 
is sufficient for our purposes. Higher precision can be reached 
within the used approach without problems: one can either increase 
the values of $-z_{l}$ and $z_{r}$ and use a higher number of 
polynomials, or use higher order terms in the asymptotic solution of 
$A$ for $z\to \pm \infty$.

\bibliographystyle{plain}

\end{document}